\documentclass[11pt]{article}
\usepackage{amsmath, amsgen,amstext,amsbsy,amsopn,amsthm, amssymb, amscd}
\usepackage{pdfsync}
\newcommand{\version}{March 27, 2014}

\usepackage{graphicx, epstopdf}

\usepackage{caption}
\usepackage{subcaption}

\usepackage{latexsym}

\newtheorem{thm}{Theorem}

\newtheorem{prop}{Proposition}
\theoremstyle{definition}

\newcommand{\be}{\begin{equation}}
\newcommand{\ee}{\end{equation}}
\newcommand{\beq}{\begin{equation}}
\newcommand{\eeq}{\end{equation}}
\newcommand{\bea}{\begin{eqnarray}}
\newcommand{\eea}{\end{eqnarray}}
\newcommand{\beqa}{\begin{eqnarray}}
\newcommand{\eeqa}{\end{eqnarray}}

\newcommand{\eps}{\varepsilon}

\def\mfr#1/#2{\hbox{$\frac{{#1} }{ {#2}}$}}

\newcommand{\simt}{\mathrel{\rlap{\hbox{$\sim$}}\raise.9ex\hbox{{\tiny
$\,$T}}}}
\newcommand{\sima}{\mathrel{\rlap{\hbox{$\sim$}}\raise.95ex\hbox{{\tiny
$\,$A}}}}

\begin{document}

\title{Entropy Meters and the  Entropy of Non-extensive Systems}
\author{Elliott H. Lieb${}^{1}$ and Jakob Yngvason${}^{2,3}$\\
\normalsize\it ${}^{1}$ Depts. of Mathematics and Physics, Princeton
University, Princeton, NJ08544, USA\\
\normalsize\it ${}^{2}$ Faculty of Physics, University of Vienna, Austria.\\
\normalsize\it ${}^{3}$ Erwin Schr{\"o}dinger Institute for Mathematical Physics, Vienna, Austria.}
\date{\version}
\maketitle 
\begin{abstract}

In our derivation of the second law of thermodynamics from the relation of
adiabatic accessibility of equilibrium states we stressed the importance of
being able to scale a system's size without changing its intrinsic
properties. This leaves open the question of defining the entropy of 
macroscopic, but unscalable systems, such as gravitating bodies or systems
where surface effects are important. We show here how the 
problem can be overcome,  in principle, with the aid of an `entropy meter'. 
An entropy meter can also be used to determine  entropy
functions for  non-equilibrium states and mesoscopic systems. 
\end{abstract}

\section{Introduction}
In our previous work \cite{LY1}--\cite{LY5} (see
also \cite{T}) we
showed how to define the
entropy of `normal systems' in equilibrium
that are scalable, and showed that this entropy is essentially unique. It
was derived without introducing the concepts of heat or temperature, and was
based solely on the notion of {\it adiabatic accessibility} and
comparability of states with respect to this relation. In a word, the
entropy of a system was defined by letting scaled copies of a system
act on each other via an {\it adiabatic process}. This procedure is
obviously not
appropriate for systems that cannot be divided into parts that have 
intrinsic properties identical to those of  a  larger system.

Here, instead, we propose to use a normal system (defined at the end of
Section \ref{two}) for which the entropy
has already been established, as an `entropy meter' by letting it act, in an
adiabatic process,
on a system whose entropy is to be determined. 
The standard way to measure entropy, as
illustrated e.g., by the `entropy meter' in \cite[pp.35-36]{A},
presupposes that the system to be measured has a well defined entropy
and, more importantly, assumes that it has a definite absolute temperature
uniformly throughout the system. The definition of temperature in a
non-normal system is not at all obvious. Our entropy meter assumes none of
these things and is based, instead, on the relation of adiabatic
accessibility,  as in \cite{LY2}. R. Giles's work \cite{G} is a precursor of
ours, as we stated in \cite{LY2}, 
but his definition
of entropy for general systems, while similar in spirit, is not the same as
the one described here or in \cite{LY5}. 
Another step in the direction of a definition of entropy for general systems has been taken by J.P. Badiali and 
A. El Kaabouchi \cite{B} who consider systems having scaling properties with fractional exponents and satisfying modifications of the axioms of \cite{LY2}.

{\it Comment: } The word `meter' as used in our paper is 
a bit unusual in the sense that the measurement involves changes in the
system to be measured, whereas a `meter' is normally thought to be 
best if it interacts least. However, any physical measurement of entropy,
for any kind of system, 
requires a 
state change, e.g., integration of $\delta Q / T$. Practically speaking,
changing the state of the sun is out of bounds, but there are many human
sized, non-scalable and non-equilibrium systems that need to be considered,
e.g., systems with sizeable surface contributions to the entropy. 

Our motivation 
is to identify entropy as a quantity that
allows us to decide which states can be transformed, adiabatically, into
which other states. Here we recall that an {\it adiabatic process} for us is
a `work process' \cite{GB, BZ} and does {\it not} require adiabatic enclosures or
slow motion or any such constraints. We do not want to introduce heat or temperature
{\em ab initio}, and thus require {\it only} that changes 
in an adiabatic process leave no mark on the 
universe other than the raising/lowering of a weight or
stretching/compressing a spring.

Our definition of entropy is presented in the next three sections for three classes of systems. In each section we define two entropy functions, denoted by $S_-$ and $S_+$, which are determined by a double variational principle. The definitions of these functions are illustrated by Figures 1, 2 and 3. An essentially unique entropy characterizing the relation $\prec$ exists if and only if these two functions are equal and this, in turn, is equivalent to the condition of {\it comparability} of the states under consideration. This comparability, which is a highly nontrivial property, was established for normal, scalable systems (called `simple systems') in \cite{LY2} by using certain structural properties of the states of normal systems that are physically motivated, but go way beyond other much simpler and almost self-evident order-theoretical assumptions about the relation $\prec$ that were the only properties used in the first part of our paper. In \cite{LY5} we argued that comparability can generally not be expected to hold for non-equilibrium states. In Section 4, where we use entropy meters to construct entropy for general, non-scalable systems, we {\it assume} comparability, but show also that should comparability not hold, the two different functions $S_\pm$ nevertheless still encode physically useful information about the relation of adiabatic accessibility.

Since our definition of entropy (or entropies) uses {\it only} the relation $\prec$ and its properties, it can be used in any situation where such a relation is given. Hence our definitions are, in principle, applicable also to mesoscopic systems and to non-equilibrium states. For the latter it provides an alternative route to the method of \cite{LY5} which is sketched in Section 3. Concerning mesoscopic systems it can be expected that the relation $\prec$, and hence the Second Law,  becomes increasingly ``fuzzy'' when the size of the system approaches atomic dimensions and the possibility of quantum entanglement between a system and its surroundings has to be taken into account. (See, e.g., \cite{Br, EDRV,  Ho}.) In such extreme situations our framework will eventually cease to apply, but there is still a wide intermediate range of sizes above atomic scales where a non-extensive entropy in the sense of the present paper may be a useful concept.

A final point to mention is that direct applications of the formulas (1)-(2), (3)-(4) or (6)-(7) may not be the most convenient  way to determine entropy  in practice, although we have shown that it is possible in principle. The {\it existence} of entropy is still a valuable piece of information, and in the cases when we have shown {\it uniqueness} we can be sure that more conventional methods, based, e.g., on  measurements of  heat capacities, compressibilities etc., will give the {\it same} result. This applies also to definitions based on formulas from statistical mechanics,  {\it provided} these can be shown to characterize the relation $\prec$. Note that although the entropy as defined in the present paper need not be extensive (the concept of `scaling' may not be applicable), it is still {\it additive} upon composition of states in cases where the comparability property holds, according
to Theorem 1 below. Additivity is not always fulfilled for entropy functions that have been proposed as generalization of the Boltzmann-Gibbs entropy in statistical physics, as, e.g., in \cite{TS} and  \cite{Ga}. A relation between states in thermodynamics characterized by such an entropy can therefore not be the same as the one considered in the present paper.

\section{Basic Definition of Entropy} \label{two}
 
We start with a very brief outline of our definition of entropy for normal
systems in \cite{LY2}. See \cite[Section 2]{LY5} for a concise summary. The
set of equilibrium states of a system of a definite amount of matter is
denoted by $\Gamma$. It is not necessary to parametrize the points of
$\Gamma$ with energy, volume, etc. for our purposes here, although we do so
in \cite{LY2} in order to derive 
other thermodynamic properties of the system, specifically temperature. 

If $X$ and $Y$ are points in two (same or different) state spaces
we write $X\prec Y$ (read `$X$ precedes $Y$') if it is possible to change $X$ to $Y$ by an adiabatic
process in the sense above. We write $X\prec \prec Y$ (read `$X$ strictly precedes $Y$') if 
$X\prec Y$ but {\it not} $Y\prec X$,  and we write $X\sima Y$ (`$X$ is adiabatically equivalent to $Y$') if $X\prec Y$ {\it and}  $Y\prec
X$. 

We say that $X$ and $Y$ are (adiabatically) {\it comparable} if $X\prec
Y$ or $Y\prec X$ holds. 

Another needed concept is the composition, or
product, of two
state spaces $\Gamma_1 \times \Gamma_2$, an element of which is simply a 
pair of states denoted $(X_1, \, X_2)$ with $X_i \in \Gamma_i$. We can
think of this product space as two macroscopic objects lying side by side
on the laboratory table, if they are not too large. 
Finally, there is the scaling of states by a real number $\lambda$,
denoted by $\lambda X$. The physical interpretation (that is, however, not needed for the mathematical proofs) is that extensive state variables like the amount of substance, energy, volume and other `work coordinates' are multiplied by $\lambda$ while intensive quantities like specific volume, pressure and temperature are unchanged.

Logic requires that we introduce a {\it `cancellation law'} into the
formalism:
\begin{itemize}
\item{\it  If $(X_1, X_2) \prec (X_1, Y_2)$ then $X_2 \prec
Y_2$.} 
\end{itemize}
\noindent In
\cite{LY2} we proved this from a stability axiom, but we can remark that it
is not really necessary to prove it since the law says that we can go from 
$X_2$ to $Y_2$ without changing the rest of the universe, which is the
definition of 
 $\prec $ in $\Gamma_2$. (See \cite[pp.22-23]{LY2} for a further discussion of this point.)

To define the entropy function on $\Gamma$
we pick two {reference points} {$X_0\prec\prec X_1$} in $\Gamma$.
Suppose  $X$ is an arbitrary state with {$X_0 \prec X\prec  X_1$} (If
$X\prec X_0$, or $X_1\prec X$, we interchange the roles of $X$ and $X_0$, or
$X_1$ and $X$, respectively.) From the assumptions about the relation
$\prec$ in \cite{LY2}, we proved that the following two functions are equal:
\beq\label{entropydef}
{S_-(X)=\sup\{\lambda'\,:\,
((1-\lambda')X_0,\lambda' X_1)\prec X\}},\eeq
\beq\label{entropydef2}
{S_+(X)=\inf\ \{\lambda''\,:\, X\prec
((1-\lambda'')X_0,\lambda'' X_1)\}.}
\eeq
Moreover, there is a $\lambda_X$ such that the $\sup$ and $\inf$ are
attained at $\lambda_X$.\footnote{If $X_1\prec\prec  X$,
then $((1-\lambda)X_0,\lambda X_1)\prec X$ has the meaning $\lambda
X_1\prec((\lambda-1)X_0, X)$ and the entropy exceeds 1. Likewise, it means
that $(1-\lambda) X_0 \prec (-\lambda X_1, X)$ if
$X\prec\prec  X_0$. See \cite{LY2}, pp.\ 27--28. }

This central theorem in \cite{LY2} provides a definition of entropy by means of a double variational principle.  An
essential ingredient for the proof that $S_-(X)=S_+(X)$ for all $X$ is the {\it comparison property} (CP):
\begin{itemize}
\item{\it Any two states in the collection of state spaces $(1-\lambda)\Gamma\times \lambda \Gamma$
with $0\leq \lambda\leq 1$ are
adiabatically
comparable.}\footnote{For $\lambda=0$ or 1 the space is simply $\Gamma$, by definition.} 
\end{itemize}
\noindent The common value $\lambda_X = S_-(X) =
S_+(X)$ is, {\it by definition,} the entropy $S(X)$ of $X$.
\begin{figure}
\centering
\begin{subfigure}{.5\textwidth}
  \centering
  \includegraphics[width=6cm]{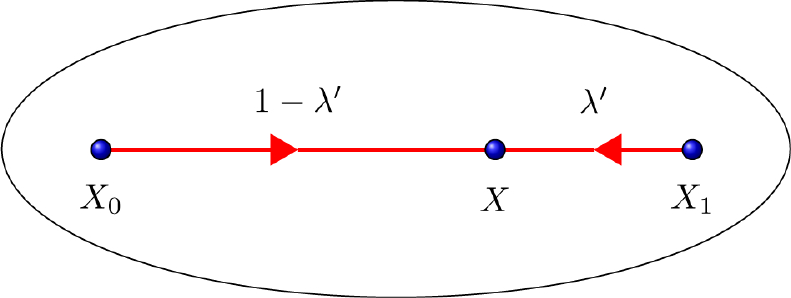}
\end{subfigure}%
\begin{subfigure}{.5\textwidth}
  \centering
  \includegraphics[width=6cm]{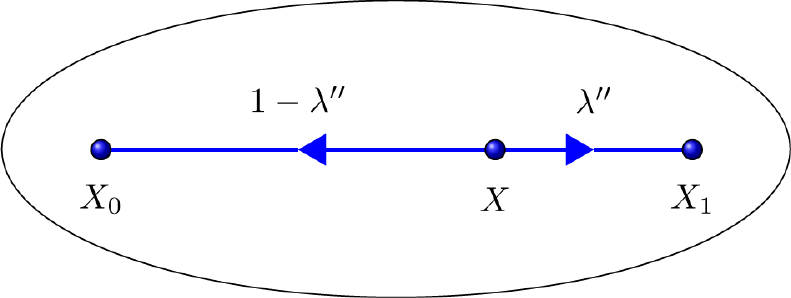}
\end{subfigure}
\caption{\small Definition of entropy for scalable systems, cf. Eqs. \eqref{entropydef} and \eqref{entropydef2}. The left figure illustrates the processes employed for definition of $S_-$, the right figure the analogous processes for $S_+$.}
\label{fig1}

\end{figure}

{\it Definition of a} {\bf normal system.} In our original paper \cite{LY2} 
we said that 
`simple systems' are the building blocks of thermodynamic systems and we
used them to prove the comparison property CP. In our work on
non-equilibrium systems \cite{LY5} we did not make use of simple systems
but we did assume, unstated, a property of such systems. Namely, that the
range of the entropy is a connected set. That is if $X,\, Y \in \Gamma$ and
$S(X) < S(Y) $ then, for every value $\lambda $ in the interval $ [S(X),
S(Y) ]$ there is a $Z
\in \Gamma$ such that 
$S(Z) = \lambda $. This property
will  be assumed here as part
of the definition of `normal systems'. The other assumptions have already
been stated, that is, the existence of an essentially unique additive
and extensive  entropy  function that characterizes the relation
$\prec$ on the state space $\Gamma$.

\section{Entropy for Non-equilibrium States of a Normal System}

In the paper \cite{LY5} we discussed the possibility of extending our
definition of entropy to  non-equilibrium states. The setting was as
follows: We assume that the space of non-equilibrium states
$\hat\Gamma$ contains a subspace of equilibrium states for which an
entropy function $S$ can be determined in the manner described above.
Moreover, we assume that the relation $\prec$ extends to $\hat\Gamma$ and 
ask for the possible extensions of the entropy from $
\Gamma $ to $\hat\Gamma$. The concept of scaling and splitting is
generally not available for $\hat\Gamma$, so that we cannot define
the entropy by means of the formulas \eqref{entropydef} and 
\eqref{entropydef2}. Instead, we
made the following 
assumption: 
\begin{itemize}
\item {\it For every $X\in \hat \Gamma$ there are $X', X''\in
\Gamma$ such that $X'\prec X\prec X''$.}
\end{itemize}
\begin{figure}[htf]
\center
\includegraphics[width=6cm]{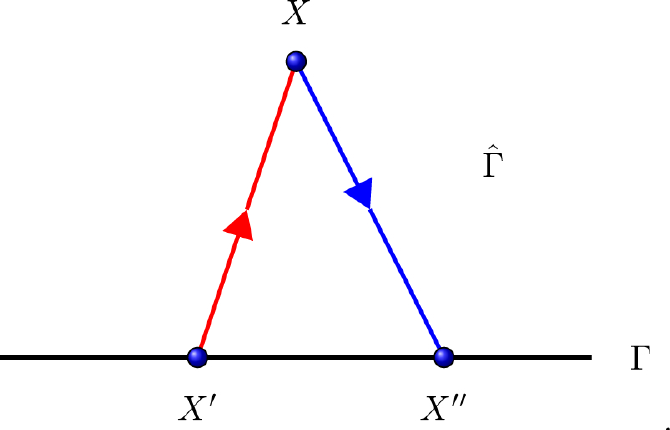}
\caption{\small The picture illustrates the definition of the entropies $S_-$, and $S_+$ for non equilibrium states of a normal system, cf. Eqs. \eqref{14} and \eqref{15}.  The space of non equilibrium states is denoted by $\hat\Gamma$  while $\Gamma$ is the subset of equilibrium states.}
\label{fig2}
\end{figure}
\noindent We then define two entropies 
for $X\in\hat\Gamma$:
\begin{equation}\label{14}
S_-(X)=\sup\{S(X')\,:\,
X'\in \Gamma, X'\prec X\}\ , 
\end{equation}
\begin{equation}\label{15}
 S_+(X)=\inf\{S(X'')\,:\,
X''\in \Gamma, X \prec X'' \}\,\, .
\end{equation}

These two functions coincide if and only if all states in $\hat\Gamma$ are
adiabatically comparable, in which case an essentially unique entropy
$S=S_-=S_+$
characterizes the relation $\prec$ on $\hat \Gamma$ in the sense that 
$X\prec Y$ if and only if $S(X)\leq S(Y)$. Whereas comparability for
equilibrium states is provable from plausible physical assumptions, however,
it is highly implausible that it holds generally for non-equilibrium states
apart from special cases, e.g. when there is local equilibrium. (See the
discussion in \cite[Section 3(c)]{LY5}.)
The functions $S_- $ and $S_+$ contain useful information, nevertheless,
because  both are monotone with respect to 
$\prec$ and every function with that property lies between $S_- $ and
$S_+$. 

\section{General Entropy Definition for Non-extensive Systems}

Our entropy meter will be a normal state space $\Gamma_0$ consisting of
equilibrium states, as in Section 2, 
with an entropy function $S$
characterizing the relation $\prec$ on this space and its scaled products.
Suppose $\prec$ is also defined on another state space
$\Gamma$,  as well as on the product of this space and
$\Gamma_0$, i.e., the space $\Gamma \times \Gamma_0$.  On such product states
the relation $\prec$ is assumed to satisfy only {\it some}  of the assumptions
that a normal space would satisfy. In the notation of \cite{LY2} these are
\begin{itemize}
{\it \item (A1) Reflexivity: $X\sima X$
\item (A2) Transitivity: $X\prec Y$ and $Y\prec Z$ implies $X\prec Z$
\item (A3) Consistency: If $X\prec X'$ and $Y\prec Y'$, then $(X,Y)\prec
(X',Y')$.
\item (A6) Stability with respect to $\Gamma_0$: If $(X,\eps Z_0)\prec
(Y,\eps Z_1)$ with $Z_0,Z_1\in\Gamma_0$ and a sequence of $\eps$'s tending
to zero, then $X\prec Y$}
\end{itemize}
Note that A4 (scaling) and A5 (splitting and recombination) are {\it not
required } for (product) states involving $\Gamma$ because the operation of
scaling need not be defined on $\Gamma$. We now pick two {\it reference
states}, $Z_0\in\Gamma_0$ and $X_1\in \Gamma$, and make the following
additional  assumption.
\begin{itemize}
\item {\it (B1) For every $X\in \Gamma$ there are $Z', Z''\in \Gamma_0$ such
that
\beq  \label{b1}
(X_1,Z')\prec (X,Z_0)\prec (X_1,Z'')
\eeq}
\end{itemize}
We use $\Gamma_0$ as an `entropy meter' to define two functions on $\Gamma$:
\beq \label{minus}
S_-(X)=\sup\{S(Z')\,:\, (X_1,Z')\prec (X,Z_0)\}
\eeq 
\beq \label{plus}
S_+(X)=\inf\{S(Z'')\,:\, (X,Z_0)\prec (X_1,Z''))\}.
\eeq
If $S_+(X) = S_-(X)$ we denote the common value by $S(X)$. 
Theorem \ref{main} will show that this is the case under a suitable
hypothesis and that $S$ has the required properties of an entropy
function. 

\noindent{\bf Remarks.}

1. The definition of $S_\pm$ is similar the one used in the proof of Theorem
2.5 in \cite{LY2} for the calibration of the multiplicative entropy
constants in products of `simple systems'.

2. The functions  defined in \eqref{minus} and \eqref{plus} give a
definition of the the upper/lower entropies of non-equilibrium 
states different from the definition given in \cite{LY5}, cf. Eqs. \eqref{14} and \eqref{15} above. Numerically, they
are identical up to additive constants, however, when both definitions apply.

\begin{figure}
\centering
\begin{subfigure}{.5\textwidth}
  \centering
  \includegraphics[width=6cm]{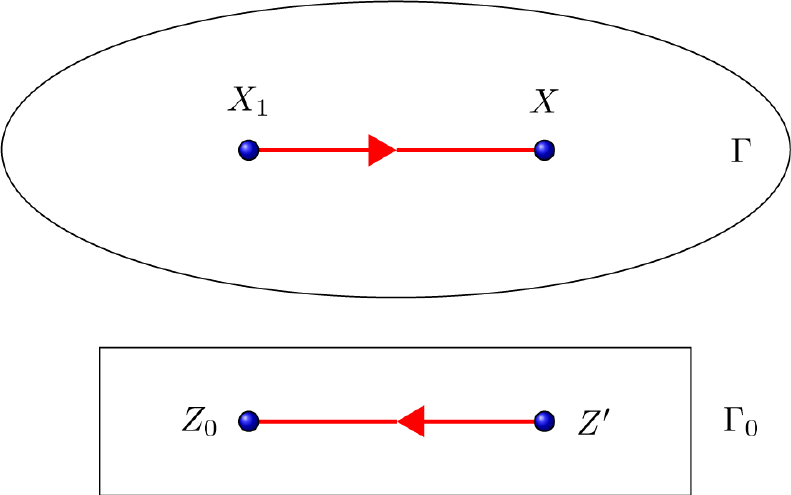}
\end{subfigure}%
\begin{subfigure}{.5\textwidth}
  \centering
  \includegraphics[width=6cm]{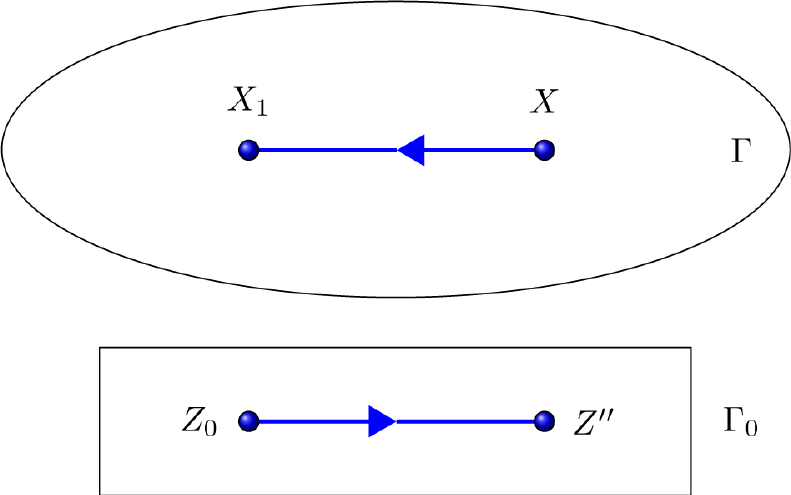}
\end{subfigure}
\caption{\small The processes used to define entropy for a system $\Gamma$ with the aid of an entropy meter, $\Gamma_0$. The left figure illustrates the definition of $S_-$ (Eq. (\ref{minus})), the right figure that of $S_+$ (Eq. (\ref{plus}).}
\label{fig3}
\end{figure}

3. Assumption (B1) may appear to be rather strong
because  when the $\Gamma$ system is large compared to the $\Gamma_0$
entropy meter then (\ref{b1}) essentially says that the small system can
move the large one from $X_1 $ to $X $ and from $X$ to $X_1$. In such a case
this can only be expected to hold if
$X$ and $X_1$ are close together. To overcome this difficulty we introduce 
`charts', as we do in differential geometry. The state space $\Gamma$
is broken into small, overlapping subregions and our Theorem
\ref{main} (with the same $\Gamma_0$ if desired) is applied to each
subregion. The saving point is that the entropy in each subregion is unique
up to an arbitrary additive constant, which means that 
the entropies in two overlapping subregions must agree up to a constant.

Can we fix an additive constant in each subregion so that every overlap
region 
has the same entropy? In principle, one could imagine an inconsistency in
the additive constants as we go along a chain of overlapping subregions.
A way to negate this possibility is to note that if one can define a global
entropy function then the mismatch along a closed loop cannot happen.
A global entropy can be constructed, in principle, however, by starting
with a sufficiently large scale copy of $\Gamma_0$, which might not be
practical physically, but which exists in principle since $\Gamma_0$ is
supposed to be scalable. With this large copy only one chart is needed and, therefore, 
the entropy exists globally. 

Our main new result is the following, which shows that $\Gamma_0$
can be used to determine, essentially uniquely, an entropy function on
the non-extensive system $\Gamma$. More generally, we can consider a
product  $\Gamma_1\times \Gamma_2 \times \dots \times \Gamma_n$ of such
non-extensive systems. 
\begin{thm}  \label{main}
Let us assume, in addition to the conditions above, 
\begin{itemize}
\item (B2) Comparability: Every state in any multiple-product  of
the spaces under consideration
is comparable to
every other state in the same multiple-product space.
\end{itemize}
Then $S_-=S_+$ and this function, denoted again by $S$, is an entropy on
$\Gamma$  in the sense that $X\prec Y$ if and only of $S(X)\leq S(Y)$.  A
change of $Z_0$ or $X_1$ amounts to a change of $S$ by an additive constant.

The entropy is additive in the sense that the function defined by $S(X,Y)=S(X)+S(Y)$, with
$X,Y\in\Gamma$, is an entropy on $\Gamma\times \Gamma$, and likewise
$S(X,Z)=S(X)+S(Z)$ with $X\in\Gamma$, $Z\in\Gamma_0$, is an entropy on
$\Gamma\times\Gamma_0$. More generally, the entropy is additive on a product
of systems
$\Gamma_1\times \Gamma_2 \times \dots \times \Gamma_n$, in the sense that
$S(X_1)+S(X_2) + \cdots +S(X_n)$ is an entropy on this space.

Finally, the entropy is
determined uniquely by these properties, up to an arbitrary additive
constant. Its unit of entropy is that of  $\Gamma_0$. 
\end{thm}

\begin{proof} {\it STEP 1:} The proof that $S_- = S_+ = S$, and that $S$ is
an entropy is similar to the proof
of Proposition 3.1 in \cite{LY5}.
We start by proving that 
for every $X\in \Gamma$ there is a $Z_X\in \Gamma_0$ such that 
\begin{equation}\label{aa}
 (X,\, Z_0) \sima  (X_1,\, Z_X).
\end{equation}
To prove \eqref{aa} we use the stability assumption (A6) for $\Gamma_0$ to
show that the $\sup$ and $\inf$ in the definitions \eqref{minus} and
\eqref{plus} are attained, that is there are $Z_{X}' $ and $Z_X''$ in 
$\Gamma_0$ such that $S_-(X) = S(Z_X') $ and $S_+(X)= S(Z_X'')$.

Indeed, since $S(Z') \leq S(Z'')$, if $Z'$ and $Z''$ are as in \eqref{b1},
and $\Gamma_0$ is a normal system, there is a $Z_X' \in \Gamma_0 $ such that $
S_-(X) = S(Z_X' )$.  We claim that $(X_1, Z_X') \prec (X,Z_0)$.   By
definition of $S_-(X)$, for every $\varepsilon >0$ there is a $Z_{\varepsilon}'
\in \Gamma_0$ such that $(X_1,Z_\eps')\prec (X,Z_0)$ and $0\leq
S(Z_X')-S(Z_\eps')\leq\eps$. Now pick two states $Z_1,Z_2\in\Gamma_0$ with
$S(Z_1)-S(Z_2)>0$. Then there is a $\delta(\varepsilon)\to 0$ such that
$S(Z_X')+\delta(\varepsilon) S(Z_1)=S(Z_\varepsilon')+\delta(\varepsilon) S(Z_2)$
which means that $(Z_X', \delta(\varepsilon) Z_1)\sima (Z_\varepsilon',
\delta(\varepsilon) Z_2)$. This in turn, implies $(X_1,Z_X',\delta(\varepsilon)
Z_1)\sima (X_1,Z_\varepsilon',\delta(\varepsilon) Z_2)\prec
(X,Z_0,\delta(\varepsilon) Z_2) $ and hence $(X_1,Z_X')\prec  (X,Z_0)$ by
stability. The existence of $Z_X''$ with $S_+(X)=S(Z_X'')$ is shown in the
same way. This establishes the existence of a maximizer in \eqref{minus} and
a minimizer in \eqref{plus}.

If $S_-(X) < S_+(X)$ there is, by the definition of normal systems, a
$\tilde{Z} \in \Gamma_0$ with $S(Z_X') < S(\tilde Z) < S(Z_X'')$. (It is here that we use the assumption of connectivity of the range of $S$.)
By comparability, we have either $(X_1, \tilde{Z })\prec (X,Z_0)$, which
would contradict $S_-(X) = S(Z_X')$ or else we have $(X, Z_0) \prec (X_1,
\tilde Z)$
which would contradict $S_+(X) = S(Z_X'')$.   Hence   
$S_-(X) = S_+(X) = S(X)$.  Either   $ Z_X' $ or $ Z_X''$ can be taken as
$Z_X$. This establishes \eqref{aa}.

Now we take $ X,\  Y \in \Gamma$. We have that  both
$(X,Z_0) \sima (X_1, Z_X)$ and $(Y,Z_0) \sima (X_1, Z_Y)$  hold,  which
implies the following equivalences:
\begin{equation}
X\prec Y \ {\rm if\ and\ only\ if\ } Z_X \prec Z_Y \ {\rm if\ and\ only\ if\
}
S(X) =S(Z_X) \leq S(Z_Y)  = S(Y).
\end{equation}
Therefore, $S$ is an entropy on $\Gamma$. 

\medskip
{\it STEP 2: }
If $\tilde Z_0$ and $\tilde X_1$ are different reference points, then likewise there is
a $\tilde Z_X$ such that 
\beq\label{b} (X,\tilde Z_0)\sima (\tilde X_1,\tilde Z_X ),\eeq
and we denote the corresponding entropy by $\tilde S(X)=S(\tilde Z_X)$.
Now (\ref{aa}) and (\ref{b}) imply
\beq \label{11}
(X_1,Z_X,\tilde Z_0)\sima (X,Z_0,\tilde Z_0)\sima (\tilde X_1,\tilde Z_X, Z_0)\sima
(X_1,Z_{\tilde X_1},\tilde Z_{X}).
\eeq
In the three steps we have used, successively, 
$(X_1,Z_X)\sima (X,Z_0)$, $(X, \tilde Z_0)\sima (\tilde X_1,\tilde Z_X)$ and $(\tilde X_1,Z_0)\sima (X_1,Z_{\tilde X_1})$.
By the cancellation law,  \eqref{11} implies
\beq (Z_X, \tilde Z_0)\sima (Z_{\tilde X_1},\tilde Z_{X})\eeq
which, because $\Gamma_0$ is a normal state space with an additive
entropy, is equivalent to
\beq 
S(X)+S(\tilde Z_0)=S(\tilde X_1)+\tilde S(X).
\eeq

\medskip
{\it STEP 3: }
The proof that $S(X)+S(Y)$ is an entropy on $\Gamma\times \Gamma$ goes as
follows:
$(X,Y)\prec (X',Y')$ is (by A3 and the cancellation property) equivalent to $(X,Y,Z_0,Z_0)\prec
(X',Y',Z_0,Z_0)$, which in turn is equivalent to $(X_1,X_1,Z_X,Z_Y)\prec
(X_1,X_1,Z_{X'},Z_{Y'})$. By cancellation this is equivalent to 
$(Z_X,Z_Y)\prec (Z_{X'},Z_{Y'})$, and by additivity of the the entropy on
$\Gamma_0\times \Gamma_0$, and by the definition of the entropies on $\Gamma$,  this
holds if and only if $S(X)+S(Y)\leq S(X')+S(Y')$. The additivity of the
entropy on $\Gamma\times \Gamma_0$, as well as on $\Gamma_1 \times \dots
\times  \Gamma_n$ is
shown in the same way.

\medskip
{\it STEP 4:}
To show that any additive entropy function $\tilde S $ on $\Gamma
\times \Gamma_0$ that satisfies the condition $\tilde S (X,Z)
=\tilde S(X)+  S(Z)$ necessarily coincides with $S(X) +S(Z) $ up to an
additive constant,
we start with \eqref{aa}, which implies $\tilde S(X) + S(Z_0) =\tilde S(X_1)
+ S(Z_X) $. However, $S(Z_X) =S(X)$, as we proved, and, therefore, 
$\tilde S(X) = S(X) + (\tilde S (X_1) -S(Z_0)$, as required. 
\end{proof}

Since the comparison property (B2) is highly nontrivial and cannot be expected to hold generally for non-equilibrium states, as we discussed in \cite{LY5}, it is important to know what can be said without it. If (B2) does not hold the functions $S_\pm$ defined in Eqs. \eqref{entropydef} and \eqref{entropydef2} will generally depend in
a non-trivial way on the choice of the reference points, and they need not be additive. They will, nevertheless, share some useful properties with the functions defined by \eqref{14} and \eqref{15}. The following Proposition is the analogue of Proposition 3.1 in \cite{LY5}:

\begin{prop}
The functions $S_\pm$ defined in Eqs.\ \eqref{entropydef}, \eqref{entropydef2} have the following properties, which do not depend on (B2):
\begin{itemize}
\item[(1)] $X\prec Y$ implies $S_-(X)\leq S_-(Y)$ and $S_+(X)\leq S_+(Y)$.
\item[(2)] If $S_+(X)\leq S_-(Y)$ then $X\prec Y$.
\item[(3)] If we take $(X_1,X_1)\in \Gamma\times\Gamma$ and $Z_0\times Z_0\in\Gamma_0\times \Gamma_0$ as reference points for defining $S_\pm$ on $\Gamma\times \Gamma$ with $\Gamma_0\times \Gamma_0$ as entropy meter, then  $S_-$ is superadditive and $S_+$ is subadditive under composition, i.e.,
\beq S_-(X)+S_-(Y)\leq S_-(X,Y)\leq S_+(X,Y)\leq S_+(X)+S_+(Y).\eeq
\item[(4)] 
If we take $(X_1,Z_0)$ and $Z_0$ as reference points
for the definitions of $S_\pm$ on $\Gamma\times \Gamma_0$, with $\Gamma_0$
as entropy meter,  then the functions $S_\pm$ on this space satisfy
\beq \label{15a}
S_\pm(X,Z_0)=S_\pm(X)\quad and \quad S_\pm(X_1,Z)=S(Z).
\eeq
If $\hat S$ is any other monotone function with respect to the relation $\prec$ on $\Gamma \times \Gamma_0$, such
that $\hat S  (X_1,Z)=S(Z)$, 
then
\beq
S_-(X) \leq \hat S(X,\,Z_0) \leq S_+(X) \quad { for \ all \ } X\in \Gamma.
\eeq
\end{itemize}
\end{prop}
\begin{proof}
PART (1). If $X\prec Y$ then, by the definition of $Z_X'$ (cf. Step 1 of the proof of Theorem 1) we have $S_-(X)=S(Z_X')$ 
and $(X_1,Z_X')\prec (X,Z_0)\prec (Y,Z_0)$. By the definition of $S_-(Y)$ this implies $S_-(X)=S(Z_X')\leq S_-(Y)$. In the same way one proves $S_+(X)\leq S_+(Y)$ by using the property of $Z_Y''$.

PART (2). If $S_+(X)\leq S_-(Y)$, then $S(Z_X'')\leq S(Z_Y')$ which implies $Z_X''\prec Z_Y'$.  Hence
$(X,Z_0)\prec (X_1, Z_X'')\prec  (X_1, Z_Y')\prec (Y,Z_0)$, and thus $X\prec Y$,  by cancellation.

PART (3). We have $(X_1, X_1, Z_X', Z_Y')\prec (X,Y, Z_0,Z_0)\prec (X_1, X_1, Z_X'', Z_Y'')$. By the definition of $S_\pm$ on $\Gamma\times \Gamma$, this implies $S( Z_X', Z_Y')\leq S_-(X,Y)\leq S_+(X,Y)\leq S( Z_X'', Z_Y'')$ and the statement follows from the additivity of $S$ on $\Gamma_0\times\Gamma_0$.

PART (4).  By definition,
\begin{equation}
 S_-(X,Z) = \sup\{S(Z)  : (X_1, Z_0,  Z') \prec (X_1,  Z_1,
Z_0)\}
= \sup \{S(Z') : (X_1,Z') \prec (X, Z) \},
\end{equation}
where the cancellation property has been used for the last equality. In
the same way,
\begin{equation}
 S_+ (X,Z) = \inf \{ S(Z'')\, :\, (X,Z)\prec (X_1,Z'')\}.
\end{equation}
This immediately implies \eqref{15a}. 

Now let $\hat S$ be monotone on $\Gamma\times \Gamma_0$, with
$\hat S(X_1,Z)= S(Z)$.  We have $S_-(X) = S(Z_X')$ with 
$(X_1,Z_X')\prec (X,Z_0)$. Therefore, $S_-(X) =S(Z_X') = \hat S (X_1,
Z_X') \leq \hat S(X,Z_0 $.  

In the same way, $\hat S(X,Z_0) \leq S_+(X).$
\end{proof}

\section{Conclusions}
We have considered the question of defining entropy for states of systems
that do not have the usual property of scalability or of being
in equilibrium, especially the former. We do so in the context of our
earlier definitions of entropy via the relation of adiabatic accessibility,
without introducing heat or temperature as primary concepts. We make no 
reference to statistical mechanical definitions but only to
processes that are physically realizable -- in principle, at least.

Our tool is an 'entropy meter', consisting of a  normal
system for which entropy has been firmly established by our previous
analysis. By measuring the change in entropy of the meter when it
interacts with the system to be measured we can, in favorable cases, define
an unambiguous entropy function for states of the observed system. We find
that  the quantity so defined actually has the properties expected of
entropy,
namely that it characterizes the relation of adiabatic accessibility (i.e.,
one state is accessible from another if and only if  its entropy is
greater), and is additive under  composition of states. 

A central concept is {\it comparability of states}, which we proved
for equilibrium states of normal systems in our earlier work. This property 
cannot be expected to hold, generally,  for non-equilibrium states, as
discussed in \cite{LY5}. 
We can, however, always define two functions. $S_-$ and $S_+$ for systems,
which have some of the properties of entropy, and which delimit the range
of possible adiabatic processes,
but it is only for the favorable case $S_- = S_+$ that a true entropy can be
proved to exist -- as we do here under the condition that comparability
holds.

\bigskip

\noindent{\textbf{Acknowledgements.}\\
Work partially supported by  U.S. National Science Foundation (grants
PHY 0965859 and 1265118; E.H.L.), 
and the Austrian Science Fund
FWF (P-22929-N16; J.Y.). We thank the Erwin Schr\"odinger Institute of the
University of Vienna
for its hospitality and support.}

 \end{document}